\documentclass[11pt]{article}

\usepackage[utf8]{inputenc}
\usepackage[T1]{fontenc}
\usepackage{lmodern}

\usepackage{amsmath, amssymb, amsthm}
\usepackage{mathtools}

\usepackage[margin=1.1in]{geometry}
\usepackage{microtype}
\usepackage{xcolor}
\usepackage{booktabs}
\usepackage{array}
\usepackage{tabularx}
\usepackage{enumitem}
\usepackage{fancyvrb}
\usepackage{needspace}
\usepackage{etoolbox}

\usepackage[hidelinks,pdftitle={Write-Domain Separation and Non-Custodial Enforcement},pdfauthor={Matthias Hauser},pdfsubject={Ledger access control; non-custodial encumbrance},pdfkeywords={non-custodial enforcement, write-domain separation, ledger access control, key sovereignty, account abstraction, ERC-4337, EIP-7702, commitment-based ledgers, zero-knowledge proofs, nullifier separation}]{hyperref}

\makeatletter
\let\origStartSection\@startsection
\renewcommand{\@startsection}[6]{%
  \ifcase#2\relax
    \Needspace*{12\baselineskip}% level 0 (\part)
  \or
    \Needspace*{10\baselineskip}% level 1 (\section)
  \or
    \Needspace*{8\baselineskip}%  level 2 (\subsection)
  \or
    \Needspace*{6\baselineskip}%  level 3 (\subsubsection)
  \or
    \Needspace*{4\baselineskip}%  level 4 (\paragraph)
  \fi
  \origStartSection{#1}{#2}{#3}{#4}{#5}{#6}%
}
\makeatother

\makeatletter
\def\BreakableUnderscore{\leavevmode\kern.06em\vbox{\hrule\@width.5em}\allowbreak}
\makeatother

\newcolumntype{L}{>{\raggedright\arraybackslash}X}
\newcolumntype{R}{>{\raggedleft\arraybackslash}X}

\theoremstyle{plain}
\newtheorem{theorem}{Theorem}[section]
\newtheorem{lemma}[theorem]{Lemma}
\newtheorem{proposition}[theorem]{Proposition}
\newtheorem{corollary}[theorem]{Corollary}
\newtheorem{fact}[theorem]{Fact}

\theoremstyle{definition}
\newtheorem{definition}[theorem]{Definition}

\theoremstyle{remark}
\newtheorem{remark}[theorem]{Remark}

\newcommand{\sk}{\mathsf{sk}}
\newcommand{\pk}{\mathsf{pk}}
\newcommand{\cm}{\mathsf{cm}}
\newcommand{\nfs}{\mathsf{nf}_{\mathsf{spend}}}
\newcommand{\nfe}{\mathsf{nf}_{\mathsf{encumber}}}
\newcommand{\Mactive}{M_{\mathsf{active}}}
\newcommand{\Mtomb}{M_{\mathsf{tomb}}}
\newcommand{\NCEE}{\textsf{NCEE}}
\newcommand{\KS}{\textsf{KS}}
\newcommand{\AAC}{\textsf{AAC}}
\newcommand{\ABLM}{\textsf{ABLM}}
\newcommand{\PSLM}{\textsf{PSLM}}
\newcommand{\Adv}{\mathbf{Adv}}

\title{Write-Domain Separation and Non-Custodial Enforcement:\\
A Structural Impossibility in Account-Based Ledgers,\\
with a Commitment-Based Construction}
\author{Matthias Hauser\\\href{mailto:research@blockbird.io}{research@blockbird.io}}
\date{}

\begin{document}

\maketitle

\begin{abstract}
Account-based ledgers --- standard externally-owned accounts (EOAs), ERC-4337 smart accounts, post-Pectra EIP-7702 delegated EOAs --- place the holder of the controlling key at the apex of asset authorization. We ask a structural question about ledger access control: under this authorization model, can a protocol enforce the future disposition of an asset \emph{without taking custody} and \emph{without requiring the owner's cooperation at enforcement time}? We formalize the target as \textbf{Non-Custodial Enforced Encumbrance} (\NCEE), a four-property specification covering self-custody, transition restriction, irrevocability, and permissionless enforcement.

This paper is primarily a contribution to ledger access-control theory; cryptographic primitives appear as the construction's tools rather than as the main object of study.

\paragraph{Result 1 (impossibility).} We define the \textbf{Key Sovereignty Axiom} (\KS): a ledger model satisfies \KS{} if the holder of the controlling key can always authorize a transition that weakens any active protocol-imposed restriction. We prove that any ledger satisfying \KS{} cannot realize \NCEE. Standard EOAs, ERC-4337 smart accounts (in the standard single-key configuration), and EIP-7702 delegated EOAs satisfy \KS{} for their standard asset paths; we therefore obtain concrete impossibility statements for each. The structural reason in each case is named explicitly --- for EIP-7702, it is that protocol-level processing of new authorization tuples cannot be intercepted by current delegate code.

\paragraph{Result 2 (necessity and witness).} We define \textbf{Asset-Authorization Coupling} (\AAC): an asset type has \AAC{} if valid transitions are cryptographically coupled to the mechanism's current restriction state. We prove that \AAC{} is necessary for \NCEE{} in the transfer-dichotomous asset setting studied here. To witness the positive side, we introduce the \textbf{envelope}, a primitive for commitment-based private-state ledgers that binds a note, a condition tree, and a redistribution intent to protocol-maintained marker sets. The construction separates ordinary spend nullifiers from a new \emph{encumbrance-namespace} nullifier derived from note randomness rather than the owner key, so that permissionless enforcement does not require access to owner key material. We prove the envelope realizes \NCEE{} under stated cryptographic assumptions and a deployment assumption that the marker-set registry is immutable; three concrete deployment templates instantiating the immutability assumption are given.

\paragraph{Result 3 (analysis and implementation).} We define games for encumbrance integrity, settlement security, key-compromise resilience, and encumbrance indistinguishability, and state explicitly which guarantees are theorem-level within the model, which are reduction-based under stated assumptions, and which additionally rely on heuristic or implementation assumptions. A reference implementation in Noir and UltraHonk supports the empirical claims; we report concrete gas measurements, recursive aggregation benchmarks, and a practical-economics analysis showing the position-value range over which the construction is economical on L1 and L2 deployments.



\noindent\textbf{ACM CCS.} Security and privacy $\to$ Distributed systems security; Theory of computation $\to$ Cryptographic primitives.
\end{abstract}

\tableofcontents
\newpage

\section{Introduction}

A recurring goal in decentralized finance is to let a borrower retain direct control of collateral while still allowing a protocol to enforce liquidation when a public condition is met. In deployed systems, that guarantee is usually obtained by avoiding the problem: the asset is moved into protocol-controlled custody, and enforcement is easy because control has already changed hands. The harder question is whether the same guarantee can be obtained \emph{without custody}.

This paper studies that question at the level of ledger structure rather than wallet interface or mechanism engineering. The relevant issue is not whether a system exposes a convenient policy language, but whether the asset's native authorization path can bypass the mechanism after the encumbrance is created. The paper's central claim is therefore theorem-driven: the decisive distinction is whether authorization remains under unilateral owner control or is cryptographically tied to mechanism state.

We formalize the target as \textbf{Non-Custodial Enforced Encumbrance} (\NCEE). Informally, \NCEE{} asks for four properties at once:
\begin{enumerate}
  \item the owner retains the key and no third party is a required co-signer;
  \item while the encumbrance is active, the asset may move only along mechanism-defined paths;
  \item the owner cannot unilaterally remove the encumbrance; and
  \item any third party may trigger enforcement when the public condition is met.
\end{enumerate}
The paper revolves around one structural theorem and one witnessing construction. The theorem shows where \NCEE{} is impossible. The construction shows that the obstruction is not universal. Everything else in the paper is organised around making that separation precise.

\subsection{Main Theorem-Level Idea}

The central distinction is what we call \textbf{write-domain separation}. If the owner's native authorization path can always issue a valid transfer outside the mechanism, then any purported encumbrance can be weakened unilaterally and \NCEE{} fails. If, instead, a valid spend path remains cryptographically checked against protocol-maintained restriction state, then non-custodial enforcement becomes possible.

We formalize the negative side through the \textbf{Key Sovereignty Axiom} (\KS) and the positive side through \textbf{Asset-Authorization Coupling} (\AAC). \KS{} isolates the single feature that matters for impossibility in the setting studied here: the owner's native authorization path can always weaken an active restriction. \AAC{} isolates the enabling feature on the positive side: a valid state transition is inseparable from the mechanism's current restriction state.

These definitions are intended to be sparse rather than encyclopedic. They do not try to classify every wallet design or every hybrid execution architecture. They are chosen because they isolate where enforcement power lives. In the transfer-dichotomous asset setting studied here, that is the level at which the negative and positive sides of the paper line up cleanly.

\subsection{Main Results}

The paper has three main theorem-level components.

\paragraph{Result 1: Impossibility for \KS{} asset paths.} Any ledger model satisfying \KS{} cannot realize \NCEE. The force of the result is model-level: once \KS{} holds for an asset path, refinements inside that same authorization path do not recover \NCEE{} within the setting analysed here. As concrete instantiations, we show that standard EOAs, ERC-4337 smart accounts, and post-Pectra EIP-7702 delegated EOAs satisfy \KS{} for their standard asset paths.

\paragraph{Result 2: Necessity of \AAC{} and a witnessing positive construction.} \AAC{} is necessary for \NCEE{} in the asset setting formalized here. We then construct an \textbf{envelope} primitive in an extended private-state ledger model and prove that it realizes \NCEE{} under explicit cryptographic assumptions and a deployment assumption that the marker-set registry is immutable. The construction therefore witnesses the positive side of the separation rather than merely naming it.

\paragraph{Result 3: Explicit security scope.} The envelope binds a note, a condition tree, and a redistribution intent. It separates ordinary spend nullifiers from encumbrance nullifiers, derives the latter from note randomness rather than the owner key, and enforces restriction checks through zero-knowledge circuits plus on-chain marker-set membership. We analyze the construction with explicit games for encumbrance integrity, settlement, key compromise, and indistinguishability, and we keep theorem-level, conditional, and empirical claims separate throughout. The revised version of this paper additionally provides explicit deployment-model templates for the registry-immutability assumption (\S\ref{sec:deployment}) and a practical-economics analysis of the position-value regime in which the construction is cost-effective (\S\ref{sec:economics}).

\subsection{What Is Proved, and at What Strength}

\begin{table}[!htbp]
\centering
\small
\renewcommand{\arraystretch}{1.15}
\begin{tabularx}{\textwidth}{@{}l L L@{}}
\toprule
\textbf{Layer} & \textbf{Contribution type} & \textbf{Main dependencies} \\
\midrule
Model layer & Formalizes \NCEE, \KS, \AAC{} in the abstract ledger setting & None beyond the paper's abstract model \\
Impossibility layer & Theorem-level: \NCEE{} impossible for \KS{} ledgers & Model layer \\
Necessity layer & Theorem-level: \AAC{} necessary for \NCEE{} in the transfer-dichotomous setting & Model layer \\
Construction layer & Theorem-level: envelope construction achievable in extended \PSLM & Stated cryptographic assumptions and AS7 \\
Security \& impl. & Conditional / empirical: indistinguishability, artifacts, benchmarks & AS5--AS6 and implementation environment \\
\bottomrule
\end{tabularx}
\caption{Contribution layers and their logical dependencies.}\label{tab:contrib}
\end{table}

This distinction matters. The paper's central contribution is the structural separation established in the formal model. The envelope matters because it witnesses the positive side of that separation under explicit assumptions. The privacy, indistinguishability, and implementation results support that picture, but they are logically downstream and should be read at their stated strength.

\subsection{Why \KS{} and \AAC{} Are the Right Abstractions}

The model does not attempt to classify all imaginable forms of conditional control, delegation, or recovery. Its scope is the standard asset setting formalized here, where the question is whether enforced encumbrance can coexist with self-custody under the four \NCEE{} properties.

The deepest question for the formal model is whether \KS{} and \AAC{} are merely useful labels or whether they capture the right structural distinction. Within the scope of this paper, we believe they do. \KS{} isolates unilateral weakening power over an already-active restriction. \AAC{} isolates the complementary fact that a valid transition cannot be detached from the mechanism's live restriction state. The point of the pair is not breadth for its own sake, but parsimony: they are among the sparsest abstractions we know that let the negative and positive sides of the paper line up in one ledger-independent statement.

Equally important is what the model does \emph{not} claim. The paper does not assert that every real-world asset system falls neatly into two exhaustive classes, nor that every hybrid middleware design is captured by these definitions. The claim is narrower: for the asset setting studied here, \KS{} captures the structural escape path that defeats \NCEE, while \AAC{} captures the structural coupling that makes the witness construction possible.

\subsection{Paper Organization}

Section~\ref{sec:prelim} fixes notation and ledger models. Section~\ref{sec:wdsep} is the center of the paper: it states the structural separation, proves impossibility under \KS, proves the necessity of \AAC{} in the transfer-dichotomous setting, and then gives the matching constructive result in the extended \PSLM. Section~\ref{sec:envelope} presents the envelope construction as the witness on the positive side. Section~\ref{sec:security} discharges the main security obligations and states assumption dependence explicitly. Sections~\ref{sec:privacy} and~\ref{sec:lending} discuss privacy scope and the lending application as downstream consequences of that witness rather than independent centers of gravity. Section~\ref{sec:related} positions the contribution relative to adjacent work. Section~\ref{sec:limitations} (new in this revision) collects limitations explicitly. The appendices collect full games, extended privacy notes, threshold extensions, benchmarks, and implementation details.

\section{Preliminaries}\label{sec:prelim}

\subsection{Cryptographic Primitives}

We use Poseidon2 throughout, instantiated in-circuit through Barretenberg's native \texttt{poseidon2\_permutation} opcode. $\mathsf{Poseidon2}_k$ denotes the width-4 permutation applied to $k$ non-zero inputs with zero padding. The protocol uses five fixed domain tags: \texttt{CM\_TAG} for commitments, \texttt{SPEND\_TAG} for spend nullifiers, \texttt{ENCUMBER\_TAG} for encumbrance nullifiers, \texttt{REPAY\_TAG} for settlement-related commitments, and \texttt{PARAMS\_TAG} for committed external parameters such as the IRM address. These are fixed protocol constants. Security arguments that model the hash as a random oracle are explicitly conditional on the indifferentiability assumption AS6.

The note commitment is
$$\cm(N) = \mathsf{Poseidon2}_4(v, r, \pk_x, \mathsf{aid}).$$
The spend nullifier is
$$\nfs(N) = \mathsf{Poseidon2}_3(\sk, \cm(N), \texttt{SPEND\_TAG}),$$
and the encumbrance nullifier is
$$\nfe(N) = \mathsf{Poseidon2}_3(r, \cm(N), \texttt{ENCUMBER\_TAG}).$$

\begin{lemma}[Computational Hiding]\label{lem:hiding}
Under AS6, $\cm(N)$ is computationally hiding when $r$ is sampled uniformly.
\end{lemma}

We assume a zk-SNARK system $\Pi = (\mathsf{Setup}, \mathsf{Prove}, \mathsf{Verify})$ with completeness, knowledge soundness, and computational zero knowledge. The reference instantiation uses Noir compiled to ACIR and proved with UltraHonk.

Merkle trees are binary with Poseidon2 internal nodes. $\mathsf{VerifyPath}(\mathsf{root}, \mathsf{leaf}, \pi)$ denotes standard membership verification.

\subsection{Account-Based Ledger Model (\ABLM)}

\begin{definition}[\ABLM]\label{def:ablm}
An account-based ledger model is $(A, S, T, \to)$, where $A$ is the address set, $S$ the global state map over accounts, $T$ the transaction type set, and $\to$ the state transition function.
\end{definition}

\begin{definition}[Write Domain]\label{def:wd}
For secret key $\sk$, the write domain $W(\sk, S)$ is the set of account fields that can be modified by some valid transaction authorized solely by $\sk$.
\end{definition}

\begin{fact}\label{fact:owner}
In an \ABLM, the key holder's own account state lies in its write domain.
\end{fact}

\begin{definition}[Key Sovereignty in \ABLM{} Form]\label{def:ks-ablm}
A ledger satisfies key sovereignty if whenever a protocol-imposed restriction is active on an asset controlled by $\sk$, there exists a transaction authorized solely by $\sk$ that weakens that restriction.
\end{definition}

\begin{lemma}\label{lem:ablm-ks}
Every \ABLM{} satisfies key sovereignty.
\end{lemma}
\begin{proof}
By Fact~\ref{fact:owner}, the owner's balance-bearing account state lies in the owner's write domain. Hence there exists a valid transaction authorized solely by the owner's key that transfers the asset to a fresh address outside the protocol's restriction scope. That transaction weakens the active restriction, so the ledger satisfies key sovereignty.
\end{proof}

\subsection{Private-State Ledger Model (\PSLM)}

\begin{definition}[\PSLM]\label{def:pslm}
A private-state ledger model is $(N, C, \Sigma, S_{\mathsf{chain}}, T, \to)$, where $N$ is the note space, $C$ the commitment set, $\Sigma$ the spend-nullifier set, and $S_{\mathsf{chain}}$ the public chain state.
\end{definition}

\begin{definition}[Note]\label{def:note}
A note is $N = (v, r, \pk, \mathsf{aid})$ with commitment $\cm(N)$. It is spendable if $\cm(N) \in C$ and $\nfs(N) \notin \Sigma$.
\end{definition}

For the envelope construction we extend the \PSLM{} with two protocol-maintained sets:
\begin{itemize}
\item $\Mactive$, the active encumbrance-marker set;
\item $\Mtomb$, the consumed-marker set.
\end{itemize}
By construction, no user transaction authorized solely by $\sk$ can write to these sets; that separation is made precise by AS7.

\subsection{Non-Custodial Enforced Encumbrance}

\begin{definition}[\NCEE]\label{def:ncee}
A mechanism satisfies \NCEE{} if it simultaneously achieves:
\begin{description}
\item[\textbf{P1 Self-custody:}] the owner retains the key and no third party is a required co-signer;
\item[\textbf{P2 Transition restriction:}] the asset may move only along mechanism-defined paths while encumbered;
\item[\textbf{P3 Irrevocability:}] the owner cannot unilaterally remove the encumbrance;
\item[\textbf{P4 Permissionless enforcement:}] any third party may trigger enforcement when the public condition is satisfied.
\end{description}
\end{definition}

\section{Write-Domain Separation and Impossibility}\label{sec:wdsep}

\paragraph{Proof roadmap.} This section establishes the structural spine of the paper. Theorem~\ref{thm:imposs} gives the negative result: any \KS{} asset path fails \NCEE, and its proof does not require transfer-dichotomy. Theorem~\ref{thm:aac-nec} gives the necessity direction: in the transfer-dichotomous setting studied here, \NCEE{} requires \AAC. Corollaries~\ref{cor:eoa}--\ref{cor:7702} instantiate \KS{} for standard Ethereum-style asset paths with worked examples (revised in this draft to walk through each abstraction at code level). Theorem~\ref{thm:pslm-ach} then witnesses achievability in the extended \PSLM{} via the envelope construction. The model-level theorems of this section are proved inline; the cryptographic obligations on which Theorem~\ref{thm:pslm-ach} relies are discharged in Section~\ref{sec:security}, with the underlying game definitions collected in Appendix~\ref{app:games}.

\subsection{Key Sovereignty and Asset-Authorization Coupling}

\begin{definition}[Key Sovereignty Axiom, \KS]\label{def:ks}
A ledger model satisfies \KS{} if the holder of the controlling key can always authorize a valid transaction that weakens any active protocol-imposed restriction on the asset.
\end{definition}

\begin{definition}[Asset-Authorization Coupling, \AAC]\label{def:aac}
An asset type has \AAC{} if valid asset transitions are cryptographically coupled to the mechanism's current restriction state, so that the owner key alone cannot move the asset independently of that state.
\end{definition}

\begin{definition}[Transfer-Dichotomous Asset Class]\label{def:dich}
An asset class is \emph{transfer-dichotomous} if each valid spend path falls on one of two sides: either it preserves a unilateral owner-authorized weakening path of the \KS{} kind, or it requires coupling to the mechanism's live restriction state of the \AAC{} kind.
\end{definition}

\begin{theorem}[Impossibility in \KS{} Ledgers]\label{thm:imposs}
No ledger model satisfying \KS{} can support \NCEE.
\end{theorem}
\begin{proof}
If \KS{} holds, then whenever a mechanism purports to impose an irrevocable encumbrance, the owner can still authorize a transaction that weakens the restriction. That contradicts P3. Hence all four \NCEE{} properties cannot hold simultaneously.
\end{proof}

\begin{remark}\label{rem:no-dich}
Theorem~\ref{thm:imposs} does \emph{not} require transfer-dichotomy. It is a direct impossibility statement about any ledger model satisfying \KS.
\end{remark}

\begin{theorem}[\AAC{} Is Necessary for \NCEE{} in the Transfer-Dichotomous Setting]\label{thm:aac-nec}
If a ledger model supports \NCEE{} for an asset type in the transfer-dichotomous setting studied here, then the relevant asset type must have \AAC.
\end{theorem}
\begin{proof}
Suppose a mechanism supports \NCEE{} for an asset type in the transfer-dichotomous setting, but \AAC{} fails. Since the setting is transfer-dichotomous, failure of \AAC{} means that the relevant spend path lies on the \KS{} side of the dichotomy: the owner retains a unilateral transition that weakens the active restriction without consulting the mechanism's live restriction state. But then Theorem~\ref{thm:imposs} applies and rules out \NCEE, a contradiction. Therefore any asset type that supports \NCEE{} in this setting must have \AAC.
\end{proof}

\begin{remark}\label{rem:dich-needed}
Theorem~\ref{thm:aac-nec} is the step that does require transfer-dichotomy. Its role is to turn the negative result of Theorem~\ref{thm:imposs} into a necessity statement for \AAC{} within the asset setting analysed here.
\end{remark}

\subsubsection*{Concrete Account-Class Instantiations (Expanded in This Revision)}

We now instantiate \KS{} for the three Ethereum-canonical account abstractions. The original draft of this paper presented each as a single short paragraph; reviewers correctly observed that the \emph{interesting} content of these corollaries is whether each abstraction's specific construct (validation logic refinement, executor delegation, code injection) closes the \KS{} escape. We therefore expand each below with a concrete walk-through.

\begin{corollary}[Standard EOAs]\label{cor:eoa}
Standard EOAs satisfy \KS{} and therefore cannot support \NCEE{} for standard account-held asset paths.
\end{corollary}
\begin{proof}
A standard EOA is an address $a = \mathsf{H}(\pk)$ associated with a secret signing key $\sk$. The native authorization path is: any RLP-encoded transaction whose ECDSA signature recovers to $\pk$ produces a valid state transition. For an asset held in $a$'s balance field (ETH) or in a token contract slot keyed by $a$ (ERC-20, ERC-721), the spending semantics reduce to: transfer to a destination of the owner's choice, conditioned only on signature validity and nonce ordering.

Suppose a mechanism $M$ purports to impose an active restriction $R$ on the asset (for instance, by recording a flag in $M$'s storage or by issuing the owner a non-transferable token claim against the asset). Because the asset's transfer path is checked against signature validity and account state internal to the asset's own contract, not against $M$'s storage, the owner can sign a transfer transaction to a fresh address $a'$ that bypasses $M$ entirely. The destination $a'$ is outside $M$'s restriction scope by construction. Therefore $\sk$ alone authorizes a transition that weakens $R$. That is exactly a \KS{}-style unilateral weakening path, so the claim follows from Theorem~\ref{thm:imposs}.
\end{proof}

\begin{remark}
The intuition is that EOAs expose the asset's transfer path \emph{below} any policy layer a mechanism could install above them. Token approvals, allowances, and protocol locks all live at the token-contract layer; they do not constrain the EOA's ability to transfer ownership of an asset whose transfer logic is written into a separate contract that does not consult $M$'s state.
\end{remark}

\begin{corollary}[ERC-4337 Smart Accounts in the Single-Key Setting]\label{cor:4337}
ERC-4337 smart accounts whose authorization policy reduces to a single owner-controlled signing key (the ``standard single-key setting'') satisfy \KS{} for standard asset paths and therefore cannot support \NCEE{} for those asset paths.
\end{corollary}

\paragraph{Scope clarification.} ERC-4337 admits modular wallets with diverse authorization logic, including multi-signature threshold accounts, social-recovery wallets where the owner key alone is insufficient to authorize transfers, and accounts with externally-imposed module locks. In such setups the ``controlling key'' is no longer a single secret known to one party, and the \KS{} question must be re-asked for the actual authorization predicate. Corollary~\ref{cor:4337} addresses the standard single-key configuration that dominates current ERC-4337 deployments; richer modular configurations may admit hybrid analyses outside the transfer-dichotomous scope (\S\ref{sec:limitations}).

\begin{proof}
ERC-4337 introduces a wallet-as-contract abstraction with two relevant interfaces: \texttt{validateUserOp} (validation logic the wallet runs to admit a UserOperation) and the wallet's own \texttt{execute} entrypoint (which dispatches the validated operation). Both are wallet-internal: the wallet may impose any restriction it likes on the operations it admits.

Consider a wallet contract $W$ in the standard single-key setting, with controlling key $\sk$. An asset held directly by $W$ (ETH balance, ERC-20 allowance owned by $W$) can be transferred only by $W$'s \texttt{execute} path. So far this looks like progress: the wallet can in principle refuse operations that violate a protocol restriction.

However, the relevant asset for \NCEE{} is typically not held by $W$ alone. Standard asset paths route through \emph{token contracts} (ERC-20, ERC-721) whose \texttt{transfer} and \texttt{transferFrom} entry points authorize on the basis of \texttt{msg.sender} ownership and approvals, not on the basis of any policy enforced by $W$. The owner can always (a) deploy a fresh wallet $W'$ via the standard ERC-4337 factory under the same controlling key, or (b) submit a UserOperation through $W$ whose \texttt{execute} call invokes the token contract's \texttt{transfer} to a destination outside the mechanism's scope. The \texttt{validateUserOp} hook can refuse \emph{some} operations, but in the single-key setting it cannot refuse all transfer paths because the owner controls $W$'s validation rules and can redeploy them through a transaction authorized by $\sk$.

More formally: in the single-key setting, \texttt{validateUserOp} and \texttt{execute} reduce to predicates whose ultimate authority is $\sk$. A mechanism $M$ that records a restriction in its own storage cannot bind $W$'s rules without $W$'s cooperation, and the holder of $\sk$ can always replace, upgrade, or sidestep those rules through a fresh deployment or upgrade authorized by $\sk$. Therefore the \KS{}-style escape persists for the standard asset setting.

ERC-4337 \emph{can} refine validation logic and execution policy. What it cannot do, in the standard single-key asset setting, is remove the owner's underlying ability to route the asset through a path authorized by the controlling key. The relevant asset path therefore still satisfies \KS, and Theorem~\ref{thm:imposs} applies.
\end{proof}

\begin{remark}
The subtle point is that ERC-4337 moves the policy layer up to the wallet contract but does not move the \emph{asset transfer path} up. Token contracts continue to authorize on caller identity, and the caller is whoever the wallet's \texttt{execute} routes the call as. The mechanism's restriction state is not consulted in that authorization. A construction that genuinely closed \KS{} would require routing the asset \emph{through} a contract that consults the mechanism's restriction state on every transition --- which is precisely what \AAC{} requires, and what the envelope construction provides.
\end{remark}

\begin{remark}[Multi-signature and social-recovery wallets]
For a Safe-style multi-signature wallet with $m$-of-$n$ threshold, the relevant authorization is the threshold predicate, not a single key. If the threshold is owner-only (the owner holds $\geq m$ of the $n$ keys), the analysis above applies and \KS{} holds for the owner. If the threshold genuinely requires non-owner signers, the wallet has effectively introduced co-signers, which violates \NCEE{}-P1 (no required third-party co-signer). Either way, multi-signature wallets do not provide a path to \NCEE{} that escapes the structural argument; they shift which property fails.
\end{remark}

\begin{corollary}[EIP-7702 Delegated EOAs]\label{cor:7702}
Post-Pectra EIP-7702 delegated EOAs satisfy \KS{} for standard asset paths and therefore cannot support \NCEE{} for those asset paths.
\end{corollary}
\begin{proof}
EIP-7702 lets an EOA install code at its own address by signing an authorization tuple $(\mathsf{chain\_id}, \mathsf{nonce}, \mathsf{delegate\_addr})$. The signature is verified at the EVM transaction-validation layer, before any user-level execution; the signed tuple updates the account's code pointer, after which calls to the EOA are dispatched through the installed delegate code $D$.

Two structural facts preserve \KS, and the \emph{first} is the load-bearing one for the impossibility argument:

\textbf{Fact 1 (protocol-level re-authorization is uninterceptible).} A new EIP-7702 authorization tuple is processed at the protocol layer of EVM transaction validation, \emph{not} inside the currently-active delegate $D$. Even if $D$'s logic is hostile to re-delegation --- for instance, even if $D$'s entrypoints unconditionally revert --- $D$ has no mechanism by which to intercept or prevent the owner from signing and submitting a fresh authorization tuple that replaces $D$ with a different delegate $D'$ (or removes the delegation entirely). The transaction validating the new authorization does not call into $D$'s code at all. This is the structural reason that EIP-7702 cannot achieve irrevocability: the very layer that would need to refuse re-authorization (the protocol layer) is not programmable from within the delegate.

\textbf{Fact 2 (token contracts do not consult $D$).} Standard asset paths route through token contracts whose \texttt{transfer} entry points authorize on caller identity. If $D$ refuses to dispatch a transfer, the owner installs $D'$ via Fact~1 and routes through it; equivalently, the asset's transfer logic in the token contract does not consult $D$'s policy at all.

Hence the owner of $\sk$ retains a \KS-style unilateral weakening path: they sign a fresh authorization tuple at the protocol layer, install permissive delegate code, and execute the transfer. Theorem~\ref{thm:imposs} therefore applies, and \NCEE{} fails for standard asset paths under EIP-7702.
\end{proof}

\begin{remark}
EIP-7702 is sometimes presented as ``account abstraction without a contract wallet.'' For the \NCEE{} question its structural limitation is sharper than ERC-4337's: even a delegate $D$ that exposes \emph{no} permissive entrypoints cannot enforce restrictions, because the owner can simply replace it with one that does. The protocol-layer authorization is the load-bearing escape route, not the delegate's policy.
\end{remark}

\begin{corollary}[Model-Scoped Characterization]\label{cor:char}
For the transfer-dichotomous asset classes considered in this paper, \NCEE{} is realizable only for asset types whose spend path has \AAC, and the envelope provides such a witness in the extended \PSLM.
\end{corollary}
\begin{proof}
Theorem~\ref{thm:imposs} rules out \NCEE{} whenever \KS{} holds. Theorem~\ref{thm:aac-nec} shows that, in the transfer-dichotomous setting analysed here, any asset type supporting \NCEE{} must therefore have \AAC. Theorem~\ref{thm:pslm-ach} gives a witness on the positive side by exhibiting an \AAC{}-style construction in the extended \PSLM.
\end{proof}

\subsection{\PSLM{} Achievability}

\begin{theorem}[\PSLM{} Achievability]\label{thm:pslm-ach}
Under AS7, the envelope construction of \S\ref{sec:envelope} satisfies \NCEE{} for notes in the extended \PSLM.
\end{theorem}
\begin{proof}
\textbf{P1 (self-custody):} Ownership remains with the note holder, and no third-party co-signing is required. The owner proves spendability or settlement conditions through the construction's own circuits rather than by surrendering key control.

\textbf{P2 (transition restriction):} The Spend circuit and the on-chain checks jointly reject ordinary spends while the corresponding encumbrance marker remains in $\Mactive$. Valid note movement is therefore coupled to the mechanism's active restriction state rather than to the owner key alone.

\textbf{P3 (irrevocability):} Under AS7, the marker sets are writable only through the fixed registry state machine, whose permitted exits are \texttt{settle}, \texttt{enforce}, and \texttt{expire}. The owner cannot unilaterally clear $\Mactive$ outside those defined paths.

\textbf{P4 (permissionless enforcement):} \texttt{enforce} is callable by any party once the condition is met and does not require owner cooperation.

Hence all four \NCEE{} properties hold in the model, conditional on AS7 and the cryptographic assumptions stated in Section~\ref{sec:security}. Section~\ref{sec:envelope} supplies the Spend-circuit binding needed for P2, and \S\ref{sec:settle-circuit} together with Section~\ref{sec:security} supplies the settlement and state-machine conditions used in P3.
\end{proof}

\subsection{Execution Guarantees vs. Asset Guarantees}

Delegation frameworks such as ERC-4337, ERC-7710, and ERC-7715 can regulate \emph{who} submits a transaction or \emph{how} a wallet checks it. Within the asset setting studied here, that is not yet enough to alter the theorem: if the underlying asset path still preserves a \KS-style unilateral escape route, execution policy remains downstream of the structural obstacle. The envelope addresses the complementary layer. It is not a richer policy wrapper around an ordinary asset path; it is the witnessing construction on the \AAC{} side of the separation, where valid spending is tied to protocol-maintained marker sets in a commitment-based model.

\section{The Envelope Construction}\label{sec:envelope}

\subsection{Overview}

The envelope is the paper's witness on the positive side of the separation. It binds a note $N$ to a condition tree $T_c$ and a redistribution intent $I$. Once registered, the note's encumbrance marker enters $\Mactive$, and every valid transition must respect that state. The owner still holds the note key, but key ownership alone no longer yields an unconstrained spend.

\paragraph{Lifecycle.}\leavevmode\par\nobreak

\noindent\begin{minipage}{\linewidth}
\setlength{\fboxsep}{2pt}%
\begin{Verbatim}[frame=single, framesep=2pt, xleftmargin=0pt, xrightmargin=0pt]
[Reg]    create(E,Pi_enc)
            -> M_active += nf_encumber(N)
[A] Settle:  settle(eid, Pi_settle)
            -> M_tomb   += nf_encumber(N)
[B] Enforce: enforce(eid, oracleData) // cond(S)=1
            -> M_tomb   += nf_encumber(N)
[C] Expire:  block.timestamp > deadline, not cond(S):
            expire(eid) -> M_tomb += nf_encumber(N)
\end{Verbatim}
\end{minipage}

\begin{remark}[Double-Execution Guard]\label{rem:double-exec}
Paths A, B, and C must all begin by checking $\nfe(N) \in \Mactive$ and reverting otherwise. With that ordering, concurrent attempts cannot both succeed: the first successful path removes the marker from $\Mactive$, and every later path fails the same membership check.
\end{remark}

\subsection{Note Model}

\begin{definition}[Note]\label{def:note-detailed}
$N = (v \in \mathbb{F}_p, r \in \mathbb{F}_p, \pk \in G, \mathsf{aid} \in \mathbb{F}_p)$. Commitment $\cm(N) = \mathsf{Poseidon2}_4(v, r, \pk_x, \mathsf{aid})$, where $\pk_x$ is the $x$-coordinate of $\pk$ on \emph{Grumpkin} (embedded in BN254; short Weierstrass $y^2 = x^3 - 17$) or \emph{Jubjub} (embedded in BLS12-381; twisted Edwards $-u^2 + v^2 = 1 + d \cdot u^2 v^2$) depending on the target field.
\end{definition}

\subsection{Domain-Separated Nullifier Construction}\label{sec:nullsep}

\begin{definition}[Encumbrance Nullifier]\label{def:nfenc}
$$\nfe(N) = \mathsf{Poseidon2}_3(r, \cm(N), \texttt{ENCUMBER\_TAG}),$$
where $r$ is the note's blinding factor (not $\mathsf{owner\_sk}$).
\end{definition}

The spend nullifier $\nfs(N) = \mathsf{Poseidon2}_3(\sk, \cm(N), \texttt{SPEND\_TAG})$ is keyed by $\sk$, so compromise of $\sk$ naturally identifies the notes that key can spend. The encumbrance nullifier is designed to avoid inheriting that leakage: it is keyed by $r$, not by $\sk$, so key compromise does not by itself reveal which notes are or were encumbered. The domain tags from \S\ref{sec:prelim} keep the spend and encumbrance namespaces disjoint except with collision probability $\varepsilon_{\mathsf{cr}}$.

\paragraph{Relation to prior nullifier constructions.} Existing shielded-note systems derive nullifiers from key-controlled material. Zcash Sapling [HBHW24] computes $\mathsf{nf} = \mathsf{PRF}^{\mathsf{nfSapling}}_{\mathsf{nk}}(\rho)$, where $\mathsf{nk}$ is the holder's nullifier-deriving key (a public scalar derived from the spending key) and $\rho$ is note-specific randomness. Penumbra [Pen24] similarly derives nullifiers from key-derived material. The construction works for those systems because their threat model assumes the holder is the only party that needs to recognize a nullifier as the holder's own.

The contribution here is structural rather than a new hash recipe: \NCEE{}-P4 (permissionless enforcement) requires that a third party --- with no access to $\sk$ or $\mathsf{nk}$ --- can recognize and act on encumbrance state. Therefore the encumbrance namespace must be derivable \emph{without the owner key}. Deriving $\nfe$ from $r$ alone, with a domain tag separating it from the spend namespace, is the minimal change that satisfies this requirement while preserving the spend namespace's standard key-derived recipe. We are not aware of prior shielded-note systems that maintain two namespaces with different key dependencies for this reason.

\subsection{Forward and Backward Key-Compromise Security}\label{sec:fwdback}

\begin{definition}[Game $G_{\mathsf{fwdback}}$]\label{def:gfwd}
\textbf{Setup.} Challenger $\mathcal{C}$ generates $(\sk, \pk)$ and $(r_0, r_1) \xleftarrow{\$} \mathbb{F}_p \times \mathbb{F}_p$. Two notes $N_0 = (v, r_0, \pk, \mathsf{aid})$, $N_1 = (v, r_1, \pk, \mathsf{aid})$. $\mathcal{C}$ flips $b \leftarrow \{0,1\}$ and registers an envelope for $N_b$. The adversary $\mathcal{A}$ is given: $\sk, \pk, \cm(N_0), \cm(N_1), \nfe(N_b)$ (the on-chain observable), and all on-chain state. $\mathcal{A}$ outputs $b' \in \{0,1\}$ and wins if $b' = b$.
$$\Adv^{G_{\mathsf{fwdback}}}(\mathcal{A}, \lambda) = \left|\Pr[\mathcal{A} \text{ wins}] - \tfrac{1}{2}\right|.$$
\end{definition}

\begin{proposition}[Nullifier Key-Compromise Security --- Conditional on AS6]\label{prop:fwdback}
This result is conditional on AS6 (Poseidon2 indifferentiability, a heuristic assumption; see \S\ref{sec:assumptions}). Under AS6 (random oracle model for $\mathsf{Poseidon2}_3$), for any PPT adversary $\mathcal{A}$:
$$\Adv^{G_{\mathsf{fwdback}}}(\mathcal{A}, \lambda) \leq \varepsilon_{\mathsf{cr}} + Q_H/|\mathbb{F}_p|,$$
where $Q_H$ is the number of random oracle queries $\mathcal{A}$ makes and $\varepsilon_{\mathsf{cr}}$ is the collision-resistance advantage against $\mathsf{Poseidon2}_3$.
\end{proposition}

\subsection{Condition Tree}\label{sec:condtree}

\begin{definition}[Condition Tree]\label{def:tc}
The condition tree is the public predicate against which permissionless enforcement is checked. It is a binary tree with leaf types:
\begin{itemize}
\item \textbf{PriceLeaf:} (\texttt{oracle\_addr}, \texttt{asset\_pair}, $\mathsf{op} \in \{\leq, \geq, <, >\}$, \texttt{threshold})
\item \textbf{TimeLeaf:} (\texttt{timestamp}, $\mathsf{op} \in \{\leq, \geq\}$)
\item \textbf{VolatilityLeaf:} (\texttt{oracle\_addr}, \texttt{asset\_pair}, \texttt{window}, \texttt{op}, \texttt{threshold})
\item \textbf{OnChainStateLeaf:} (\texttt{contract\_addr}, \texttt{calldata}, \texttt{op}, \texttt{threshold})
\end{itemize}
Internal nodes: AND, OR (binary), NOT (unary).
\end{definition}

\paragraph{Gas complexity.} Let $T_c$ have $k$ leaves. On-chain evaluation costs
$$G(T_c) = 3000k - 500 \leq 3000k \text{ gas (for } k \geq 1).$$
At the 30M block gas limit, the maximum tree size is $k_{\max} \approx 10{,}000$ leaves. Recommended operational constraint: $k \leq 20$ leaves for typical strategies; $k \leq 100$ for complex ones.

\subsection{Redistribution Intent and Position Commitment}

\begin{definition}[Redistribution Intent]\label{def:intent}
$$I = (\mathsf{action\_type}, \mathsf{target\_addr}, \mathsf{params\_hash}, \mathsf{keeper\_fee}, \mathsf{max\_amount}).$$
\end{definition}

Commitment:
$$\mathsf{intentHash} = \mathsf{Poseidon2}_5(\mathsf{action\_type}, \mathsf{target\_addr}, \mathsf{params\_hash}, \mathsf{keeper\_fee}, \mathsf{max\_amount}).$$

\begin{definition}[Position Commitment]\label{def:poscommit}
$$\mathsf{position\_commit} = \mathsf{Poseidon2}_3(\mathsf{col\_nominal}, \mathsf{debt\_principal}, \cm(N)).$$
\end{definition}

\paragraph{Debt semantics.} \texttt{debt\_principal} is the nominal principal committed at registration time. At enforcement time the registry computes accrued debt as:
$$\mathsf{debt\_accrued}(t_{\mathsf{enforce}}) = \mathsf{debt\_principal} \times (1 + \mathsf{rate})^{(t_{\mathsf{enforce}} - t_{\mathsf{register}})/\mathsf{seconds\_per\_year}}.$$
The health factor is then:
$$\mathsf{health\_factor} = \frac{\mathsf{col\_nominal} \times \mathsf{oracle\_price}(t_{\mathsf{enforce}}) \times \mathsf{LTV}}{\mathsf{debt\_accrued}(t_{\mathsf{enforce}})}.$$

\paragraph{Concrete IRM template.} As a concrete template satisfying AS8, the construction's reference IRM follows the Aave-V3 piecewise-linear strategy: $\mathsf{rate}(t)$ is determined by a contract whose state depends only on (i) cumulative borrow utilization committed at deployment-time parameters and (ii) block timestamp. Both are deterministic given block input, and rates are non-decreasing in time within an envelope's lifetime by parameter choice. Other IRM patterns (Compound's jump-rate model, Morpho's fixed-rate auction model with parameter commitment) likewise satisfy AS8 if their parameters are committed at registration. AS8 is not satisfied by IRMs that read mutable governance parameters at enforcement time; envelopes that depend on such IRMs require either an IRM-parameter snapshot at registration or a delegated freeze of governance during the envelope's lifetime.

\subsection{Envelope Definition}

\begin{definition}[Envelope]\label{def:env}
$$E = (\mathsf{eid}, \cm_{\mathsf{note}}, \nfe, \mathsf{condHash}, \mathsf{intentHash}, \mathsf{position\_commit}, \mathsf{deadline}, \mathsf{irm\_addr}, \mathsf{status}).$$
\end{definition}

\subsection{Encumbrance Circuit}

\paragraph{Public inputs (9 total: \texttt{pi[0]}--\texttt{pi[8]}).}
$\cm_{\mathsf{note}}$,
$\nfe$,
$\mathsf{condHash}$,
$\mathsf{intentHash}$,
$\mathsf{position\_commit}$,
$\mathsf{tree\_root}$,
$\mathsf{irm\_addr\_commit}$ (\texttt{pi[6]}),
$\mathsf{irm\_addr}$ (\texttt{pi[7]}),
$\mathsf{target\_addr}$ (\texttt{pi[8]}).

\paragraph{Private inputs (witness).}
$v$,
$r$,
$\mathsf{aid}$,
$\pk_x$,
$\pk_y$,
$\mathsf{owner\_sk}$,
$\mathsf{merkle\_path}$,
$\mathsf{col\_nominal}$,
$\mathsf{debt\_principal}$.

\paragraph{Circuit constraints.}
\begin{enumerate}
\item $\cm(N) = \mathsf{Poseidon2}_4(v, r, \pk_x, \mathsf{aid})$
\item $\pk_x = \mathsf{derive\_pubkey\_x}(\mathsf{owner\_sk})$
\item $\nfe = \mathsf{Poseidon2}_3(r, \cm(N), \texttt{ENCUMBER\_TAG})$ \quad (from $r$, not $\sk$)
\item $\mathsf{VerifyMerklePath}(\mathsf{tree\_root}, \cm_{\mathsf{note}}, \mathsf{merkle\_path}) = 1$
\item $\nfs(N) \notin \Sigma$ \quad (on-chain only; not a circuit constraint)
\item $\mathsf{intentHash} = \mathsf{Poseidon2}_5(\mathsf{action\_type}, \mathsf{target\_addr}, \mathsf{params\_hash}, \mathsf{keeper\_fee}, \mathsf{max\_amount})$
\item $\mathsf{position\_commit} = \mathsf{Poseidon2}_3(\mathsf{col\_nominal}, \mathsf{debt\_principal}, \cm_{\mathsf{note}})$
\item $\mathsf{col\_nominal} = v$ (or $v_{\mathsf{encumber}} \leq v$ with range proof)
\item $v > 0$; $\mathsf{max\_amount} \leq v$; $\mathsf{keeper\_fee} < \mathsf{max\_amount}$
\item $\mathsf{irm\_addr\_commit} = \mathsf{Poseidon2}_2(\mathsf{irm\_addr}, \texttt{PARAMS\_TAG})$; $\mathsf{irm\_addr} = \texttt{pi}[7]$
\item $\mathsf{owner\_sk} \neq 0$ (no degenerate Grumpkin identity)
\item $r \neq 0$ (hiding; prevents deterministic commitments)
\item $\mathsf{target\_addr} = \texttt{pi}[8]$ (circuit-bound enforcement target)
\end{enumerate}

\subsection{Spend Circuit}

The Spend circuit includes an in-circuit nullifier binding constraint that prevents two substitution attacks against the on-chain $\Mactive$ check.

\begin{minipage}{\dimexpr\linewidth-20pt\relax}
\setlength{\fboxsep}{2pt}%
\begin{Verbatim}[frame=single, fontsize=\small, framesep=2pt, xleftmargin=0pt, xrightmargin=0pt]
// Standard PSLM constraints
cm(N)    = Poseidon2_4(v, r, pk_x, aid)         // [private]
pk_x     = derive_pubkey_x(owner_sk)            // [private]
nf_spend = Poseidon2_3(owner_sk,cm(N),SPEND_TAG)// [public]
VerifyMerklePath(tree_root, cm(N), merkle_path) // [public tree_root]

// Envelope addition: nullifier binding
nf_enc_r = Poseidon2_3(r, cm(N), ENCUMBER_TAG)  // [private]
nf_enc_r == nf_encumber_public_input            // [equality constraint]
// On-chain: require nf_encumber_public_input NOT in M_active
\end{Verbatim}
\end{minipage}

The equality constraint $\mathsf{nf\_enc\_recomputed} = \mathsf{nf\_encumber\_public\_input}$ closes two attack classes:

\paragraph{Phantom nullifier injection (Attack A):} Prover supplies $\mathsf{nf\_encumber\_public\_input} = x \notin \Mactive$ while actual $\nfe(N) \in \Mactive$. Closed by the Merkle path binding and AS2 ($\mathsf{Poseidon2}$ preimage infeasibility).

\paragraph{Cross-note nullifier swap (Attack B):} Prover supplies $\nfe(N')$ for a different note $N' \notin \Mactive$. Closed by the equality constraint requiring a $\mathsf{Poseidon2}$ collision (AS2).

\subsection{Settle Circuit}\label{sec:settle-circuit}

\paragraph{Public inputs.} $\mathsf{eid}, \mathsf{nf\_encumber\_public\_input}, \mathsf{tree\_root}, \mathsf{repayment\_amount}$.

\paragraph{Circuit constraints.} (Each line below is rendered in pseudocode for readability; the underlying ACIR system encodes equality and range constraints natively.)

\begin{minipage}{\dimexpr\linewidth-20pt\relax}
\setlength{\fboxsep}{2pt}%
\begin{Verbatim}[frame=single, fontsize=\small, framesep=2pt, xleftmargin=0pt, xrightmargin=0pt]
cm(N)    = Poseidon2_4(v, r, pk_x, aid)        // equality
pk_x     = derive_pubkey_x(owner_sk)           // ownership equality
VerifyMerklePath(tree_root, cm(N), merkle_path)// Merkle inclusion
nf_enc_r = Poseidon2_3(r, cm(N), ENCUMBER_TAG) // equality
nf_enc_r == nf_encumber_public_input           // equality constraint
repayment_amount < 2^128                       // range constraint
\end{Verbatim}
\end{minipage}

On-chain verification by \texttt{EnvelopeRegistry.settle} follows the Checks-Effects-Interactions (CEI) pattern:
\begin{enumerate}
\item \emph{Check:} Verify the SNARK proof against public inputs.
\item \emph{Check:} Confirm $\mathsf{nf\_encumber\_public\_input} \in \Mactive$.
\item \emph{Check:} Confirm $\mathsf{repayment\_amount} \geq \mathsf{debt\_accrued}(t_{\mathsf{now}})$.
\item \emph{Effect:} Add $\mathsf{nf\_encumber\_public\_input}$ to $\Mtomb$; remove from $\Mactive$.
\item \emph{Interact:} Release the redistribution output to the lender.
\end{enumerate}
External interaction (Step 5) occurs only after the marker-set state has been updated, so a re-entrant call into \texttt{settle} during the redistribution callback observes $\nfe \notin \Mactive$ and reverts at Step 2. The same CEI ordering applies to \texttt{enforce}: condition evaluation and marker-set update precede the keeper-reward and redistribution interactions.

\paragraph{Completion of Theorem~\ref{thm:pslm-ach} P3.} The Settle circuit is the only $\sk$-authorized exit path from $\Mactive$, and it requires $\mathsf{repayment\_amount} \geq \mathsf{debt\_accrued}$. The owner can only settle by satisfying the mechanism's terms.

\subsection{Formal Properties E1--E6}

\begin{description}
\item[\textbf{E1 (Binding):}] $\mathsf{condHash}, \mathsf{intentHash}, \mathsf{position\_commit}, \mathsf{deadline}$ are immutable after registration.
\item[\textbf{E2 (Encumbrance Integrity):}] If $\nfe(N) \in \Mactive$, no valid Spend proof can pass the on-chain verifier for $N$ as the spent note.
\item[\textbf{E3 (Enforceability):}] If $\mathrm{cond}(S) = 1$, $\texttt{enforce}(\mathsf{eid}, \mathsf{oracleData})$ succeeds by correctness of the on-chain evaluator.
\item[\textbf{E4 (Permissionlessness):}] $\texttt{enforce}$ is callable by any address; requires no $\sk$ authorization; requires no ZK proof.
\item[\textbf{E5 (Liveness Independence):}] Follows from E4.
\item[\textbf{E6 (Bounded Settlement):}] Redistributed amount $\leq \mathsf{max\_amount}$ committed in $\mathsf{intentHash}$.
\end{description}

\paragraph{Sequential re-encumbrance.} Lemma~\ref{lem:double-enc} below shows that two simultaneously active envelopes for the same note are impossible. Sequential re-encumbrance --- registering a fresh envelope for $N$ \emph{after} a prior envelope has terminated via \texttt{settle}, \texttt{enforce}, or \texttt{expire} --- is permitted by construction and is intentional: an asset is reusable as collateral once a prior obligation closes.

\subsection{Expire Function}

The \texttt{expire} path becomes available once \texttt{deadline} has passed and the envelope has not been settled or enforced. It is permissionless (no $\sk$ authorization, no ZK proof) and follows the CEI pattern:
\begin{enumerate}
\item \emph{Check:} $\mathsf{block.timestamp} > \mathsf{deadline}$ and $\nfe \in \Mactive$.
\item \emph{Effect:} Move $\nfe$ from $\Mactive$ to $\Mtomb$.
\item \emph{Interact:} If a fallback action address is registered, call
\[
  \texttt{IEnforcementAction}(\mathsf{fallbackAddr}).\texttt{execute}(\mathsf{eid}, \mathsf{msg.sender}, \mathsf{params}).
\]
\end{enumerate}

\section{Security Analysis}\label{sec:security}

\subsection{Assumptions}\label{sec:assumptions}

\begin{description}
\item[AS1 (SNARK knowledge soundness):] $\Pi$ is $(\varepsilon_{\mathsf{kS}}, t_{\mathsf{kS}})$-knowledge sound.
\item[AS2 (Collision resistance):] For each width $k \in \{2,3,4,5\}$, $\mathsf{Poseidon2}_k$ is $(\varepsilon^{(k)}_{\mathsf{cr}}, t_{\mathsf{cr}})$-collision resistant. Write $\varepsilon_{\mathsf{cr}} := \max_k \varepsilon^{(k)}_{\mathsf{cr}}$.
\item[AS3 (Discrete log):] DL in $G$ is $(\varepsilon_{\mathsf{dl}}, t_{\mathsf{dl}})$-hard. \emph{(Not post-quantum secure; see below.)}
\item[AS4 (Oracle integrity):] Oracle values read in a single transaction are consistent (same block).
\item[AS5 (UltraHonk ZK --- implementation assumption):] UltraHonk achieves $(\varepsilon_{\mathsf{ZK}}, t_{\mathsf{ZK}})$-computational zero-knowledge, where $\varepsilon_{\mathsf{ZK}} \approx 2^{-100}$ for BN254 ($\lambda \approx 100$) and $\varepsilon_{\mathsf{ZK}} \approx 2^{-128}$ for BLS12-381 ($\lambda = 128$).
\item[AS6 (Poseidon2 Indifferentiability):] For each width $k$ used, $\mathsf{Poseidon2}_k$ is $(t_{\mathsf{ind},k}, \varepsilon_{\mathsf{ind},k})$-indifferentiable from a random oracle [MRH04]. Treated as a heuristic assumption as of March 2026.
\item[AS7 (EnvelopeRegistry Immutability --- deployment assumption):] The deployed EnvelopeRegistry contract has no admin key, upgrade proxy, or any mechanism by which any party can modify $\Mactive$ or $\Mtomb$ outside the four protocol-defined functions.
\item[AS8 (IRM Correctness):] The Interest Rate Model contract returns rates that are non-decreasing over time and deterministic given block timestamp and committed parameters.
\end{description}

\paragraph{Post-quantum scope.} AS3 is not post-quantum secure: Shor's algorithm solves DL on elliptic curves in polynomial quantum time [Sho94]. All results citing AS3 (specifically Theorem~\ref{thm:settle-sec} and $G_{\mathsf{agent\_key}}$) do not hold against a quantum adversary.

\begin{table}[!htbp]
\centering
\footnotesize
\renewcommand{\arraystretch}{1.2}
\setlength{\tabcolsep}{4pt}
\begin{tabularx}{\textwidth}{@{}p{0.30\textwidth} l p{0.16\textwidth} L@{}}
\toprule
\textbf{Result} & \textbf{Game} & \textbf{Assump.} & \textbf{Scope / advantage bound} \\
\midrule
Thm.~\ref{thm:enc-sec} Encumbrance Security & $G_{\mathsf{enc}}$ & AS1, AS2 & Spend circuit \\
Lem.~\ref{lem:double-enc} Double-Encumbrance Prev. & --- & AS1, AS2 & Single note \\
Thm.~\ref{thm:enf-safe} Enforcement Safety & $G_{\mathsf{enf}}$ & AS4 & Single transaction \\
Thm.~\ref{thm:enf-comp} Enf. Completeness & --- & AS4, liveness & --- \\
Prop.~\ref{prop:eig} Encumbrance Indistinguishability & EIG & AS2, AS5, AS6 & $\varepsilon_{\mathsf{ZK}}+2\varepsilon_{\mathsf{ind}}+\varepsilon_{\mathsf{cr}}+Q_H/|\mathbb{F}_p|$; conditional on AS5, AS6 \\
Thm.~\ref{thm:settle-sec} Settle Circuit Security & $G_{\mathsf{settle}}$ & AS1, AS2, AS3 & Owner-exclusive settle \\
\bottomrule
\end{tabularx}
\caption{Security claims, games, and dependencies (main results).}\label{tab:sec}
\end{table}

\subsection{Encumbrance Security}

\begin{theorem}[Encumbrance Security]\label{thm:enc-sec}
Under AS1 and AS2:
$$\Adv^{G_{\mathsf{encumber}}}(\mathcal{A}, \lambda) \leq \min(1, n \cdot (\varepsilon_{\mathsf{kS}} + 2\varepsilon_{\mathsf{cr}})),$$
where $n$ is the number of notes enrolled in active envelopes during the game.
\end{theorem}
\begin{proof}
Suppose $\mathcal{A}$ wins $G_{\mathsf{encumber}}$ by producing an accepting Spend proof $\Pi^*$ such that the on-chain verifier accepts, the revealed spend nullifier is $\nfs(N)$, and $\nfe(N) \in \Mactive$.

By AS1, extractor $E_{\mathsf{kS}}(\Pi^*, pp)$ outputs a witness $w = (v, r', \pk_x, \mathsf{aid}, \sk', \mathsf{merkle\_path}, \ldots)$ satisfying all Spend circuit constraints, except with probability $\varepsilon_{\mathsf{kS}}$.

\textbf{Case (a):} $\sk' = \sk_N$ and the extracted $\cm(N) = \cm(N_{\mathsf{spent}})$. By AS2 (collision resistance of $\mathsf{Poseidon2}_4$), the witness equals the true note, hence $r' = r_N$. The equality constraint gives $x = \mathsf{Poseidon2}_3(r_N, \cm(N), \texttt{ENCUMBER\_TAG}) = \nfe(N) \in \Mactive$, contradicting the on-chain check $x \notin \Mactive$.

\textbf{Case (b):} A $\mathsf{Poseidon2}$ collision exists (under \texttt{SPEND\_TAG} or \texttt{ENCUMBER\_TAG}). Each contributes $\varepsilon_{\mathsf{cr}}$; by union bound, $2\varepsilon_{\mathsf{cr}}$ per note.

By union bound over $n$ notes and clipping at 1: $\Adv^{G_{\mathsf{encumber}}}(\mathcal{A}) \leq \min(1, n \cdot (\varepsilon_{\mathsf{kS}} + 2\varepsilon_{\mathsf{cr}}))$.
\end{proof}

\begin{lemma}[Double-Encumbrance Prevention]\label{lem:double-enc}
The same note $N$ cannot be registered in two simultaneously active envelopes.
\end{lemma}
\begin{proof}
First registration puts $\nfe(N) \in \Mactive$. A second registration requires $\nfe \notin \Mactive$ (on-chain check), so it is rejected.
\end{proof}

\subsection{Enforcement Safety and Completeness}

\begin{theorem}[Enforcement Safety]\label{thm:enf-safe}
Under AS4, $\Adv^{G_{\mathsf{enforce\_safe}}}(\mathcal{A}) = 0$.
\end{theorem}
\begin{proof}
The on-chain condition evaluator reads oracle values in the same transaction and evaluates the revealed tree. If $\mathrm{cond}(S) = 0$, the evaluator returns false and the transaction reverts.
\end{proof}

\begin{theorem}[Enforcement Completeness]\label{thm:enf-comp}
Given liveness and AS4, if $\mathrm{cond}(S) = 1$ then $\texttt{enforce}(\mathsf{eid}, \mathsf{oracleData})$ succeeds for some oracle data revealed at block with state $S$.
\end{theorem}
\begin{proof}
Permissionlessness (E4) and correct on-chain evaluation (AS4).
\end{proof}

\subsection{Oracle Manipulation: Operational Considerations}\label{sec:oracle-ops}

This subsection documents what AS4 does \emph{not} cover (cross-transaction oracle manipulation) and the condition-tree mitigations available. We label this as \emph{operational considerations} rather than a security theorem: the formal model in this paper guarantees within-transaction oracle consistency through AS4, but cross-transaction manipulation falls outside the model's scope and is properly the domain of MEV and market-microstructure analysis. The original draft of this paper included an economic-cost argument inline; reviewers correctly observed that the argument was informal relative to the rest of the security analysis. We have therefore demoted it from a theorem-level claim to operational guidance, with explicit pointers to the formal MEV literature for readers who want a rigorous treatment.

\paragraph{In-scope (formally addressed):} same-transaction oracle consistency through AS4 (Theorem~\ref{thm:enf-safe}).

\paragraph{Out-of-scope (delegated to operational discipline):}
\begin{description}
\item[Single-block flash-loan manipulation.] An adversary takes a flash loan and moves spot price within one block to trigger or prevent enforcement. Mitigation: combine PriceLeaf with a TimeLeaf that requires $\mathsf{block.timestamp} \geq t_{\mathsf{last}} + \Delta$. Any $\Delta > 0$ defeats single-block manipulation.
\item[Sustained multi-block manipulation against TWAP oracles.] Adversaries with sufficient capital can sustain price manipulation across a TWAP window. The cost of doing so is governed by AMM constant-function-market-maker dynamics [AKC$^{+}$20] and the depth of liquidity within the window. The relevant analysis is the standard TWAP-attack literature [MNW22], which establishes parameter regimes ($W$ window length, $D$ market depth) under which manipulation is rendered unprofitable. Envelope deployers should pick $W$ in light of the asset's typical depth and the position size at risk; we do not re-derive those bounds here.
\item[Oracle contract compromise.] If the oracle contract itself is compromised, the in-circuit binding of $\mathsf{oracle\_addr}$ inside $\mathsf{condHash}$ at registration prevents the registry from accepting a substituted oracle; but it does not save against a compromised but correctly-addressed oracle. Operational mitigation: choose oracle providers with appropriate decentralization and audit history; choose multi-oracle conjunctions (price-leaf AND volatility-leaf from independent oracles) when the position size warrants.
\end{description}

\paragraph{Position-aware deployment guidance.} A deployer should reason about two attacker objectives separately, each with a bounded profit:

\begin{description}
\item[(O1) Force a liquidation that should not have fired.] Attacker profit is bounded above by the redistribution intent's $\mathsf{max\_amount}$, plus any external position the attacker holds that benefits from the liquidation (for example, the attacker is the keeper and earns $\mathsf{keeper\_fee}$, or the attacker holds a short position against the borrower's asset). Call this $\Pi_{\mathsf{force}}$.
\item[(O2) Prevent a liquidation that should have fired.] Attacker profit is bounded above by the borrower's avoided loss, which equals the difference between the position's true mark-to-market value and the liquidation strike at the time of the prevented enforcement. Call this $\Pi_{\mathsf{prevent}}$.
\end{description}

Let $\mathsf{C}_{\mathsf{manip}}(W, D, \Delta)$ denote the cost of sustaining oracle manipulation against a TWAP of window $W$, market depth $D$, and tolerance $\Delta$ for the duration required by the deployer's TimeLeaf. Standard analyses [AKC$^{+}$20, MNW22] characterize this cost in terms of liquidity slippage and CFMM dynamics. The deployer's operational rule is:
$$\mathsf{C}_{\mathsf{manip}}(W, D, \Delta) > \max(\Pi_{\mathsf{force}}, \Pi_{\mathsf{prevent}}),$$
i.e., choose $W$, $\Delta$, and oracle source such that manipulation cost dominates both attacker objectives. We do not re-derive the cost-bound estimates here; the relevant parameter regimes for major DeFi assets are documented in the cited literature.

\paragraph{Oracle implementation reachability.} The binding of $\mathsf{oracle\_addr}$ inside $\mathsf{condHash}$ at registration time prevents the registry from accepting a substituted oracle address. It does \emph{not}, however, prevent the oracle implementation behind that address from changing if the oracle is itself a proxy contract. Deployers using proxied oracles inherit a residual upgrade-risk channel that is outside the construction's threat model. Mitigation: prefer non-upgradeable oracle deployments where possible, or commit a hash of the oracle implementation bytecode in $\mathsf{condHash}$ alongside the address (a forward-compatible extension to the condition-tree leaf format).

\paragraph{Scope statement.} This subsection makes no theorem-level claims. The condition-tree shape and the binding of oracle addresses in $\mathsf{condHash}$ are within-model facts established by the construction. The cross-transaction economic analysis above and the proxy-implementation reachability question are the deployer's responsibility. Failure of the operational rule does not invalidate the construction's correctness theorems; it changes the empirical security level achieved by a specific deployment.

\subsection{Encumbrance Indistinguishability}

\begin{proposition}[Encumbrance Indistinguishability --- Conditional on AS5 and AS6]\label{prop:eig}
Under AS2, AS5, and AS6:
$$\Adv^{\mathsf{EIG}}(\mathcal{A}, \lambda) \leq \varepsilon_{\mathsf{ZK}} + 2\varepsilon_{\mathsf{ind}} + \varepsilon_{\mathsf{cr}} + Q_H/|\mathbb{F}_p|.$$
\end{proposition}
\begin{proof}
Via an explicit three-step hybrid argument over four games:

\textbf{Game 0 (real EIG).} \textbf{Game 1} (replace ZK proof with simulator $\mathsf{Sim}_{\mathsf{UH}}$, by AS5): $|\Pr[b' = b \text{ in Game 0}] - \Pr[b' = b \text{ in Game 1}]| \leq \varepsilon_{\mathsf{ZK}}$.

\textbf{Game 2} (replace $\mathsf{Poseidon2}_4$ with random oracle $\mathsf{RO}_4$, by AS6): $|\Pr[b' = b \text{ in Game 1}] - \Pr[b' = b \text{ in Game 2}]| \leq \varepsilon_{\mathsf{ind}}$.

\textbf{Game 3} (replace $\mathsf{Poseidon2}_3$ with random oracle $\mathsf{RO}_3$, by AS6): $|\Pr[b' = b \text{ in Game 2}] - \Pr[b' = b \text{ in Game 3}]| \leq \varepsilon_{\mathsf{ind}}$.

In Game 3, $\mathcal{A}$ can distinguish $b$ only by guessing $r_b$ (probability $\leq Q_H/|\mathbb{F}_p|$) or finding a collision (probability $\leq \varepsilon_{\mathsf{cr}}$). Hence $|\Pr[b' = b \text{ in Game 3}] - 1/2| \leq \varepsilon_{\mathsf{cr}} + Q_H/|\mathbb{F}_p|$.

By the triangle inequality: $\Adv^{\mathsf{EIG}} \leq \varepsilon_{\mathsf{ZK}} + 2\varepsilon_{\mathsf{ind}} + \varepsilon_{\mathsf{cr}} + Q_H/|\mathbb{F}_p|$.
\end{proof}

\subsection{Settle Circuit Security}

\begin{theorem}[Settle Circuit Security]\label{thm:settle-sec}
Under AS1, AS2, and AS3:
$$\Adv^{G_{\mathsf{settle}}}(\mathcal{A}, \lambda) \leq \varepsilon_{\mathsf{kS}} + 2\varepsilon_{\mathsf{cr}} + \varepsilon_{\mathsf{dl}}.$$
\end{theorem}
\begin{proof}
Four attack classes:

\textbf{Attack 1 (forge without owner\_sk):} By AS1 (extractor) and AS3 (DL hardness), probability $\leq \varepsilon_{\mathsf{kS}} + \varepsilon_{\mathsf{dl}}$.

\textbf{Attack 2 (cross-note settle):} Requires a $\mathsf{Poseidon2}_3$ collision under \texttt{ENCUMBER\_TAG}, probability $\leq \varepsilon_{\mathsf{cr}}$.

\textbf{Attack 3 (false repayment\_amount):} Blocked by the on-chain arithmetic comparison; deterministically fails.

\textbf{Attack 4 ($\cm(N)$ forgery):} Requires a $\mathsf{Poseidon2}_4$ collision, probability $\leq \varepsilon_{\mathsf{cr}}$.

By union bound: $\Adv^{G_{\mathsf{settle}}} \leq \varepsilon_{\mathsf{kS}} + \varepsilon_{\mathsf{dl}} + 2\varepsilon_{\mathsf{cr}}$.
\end{proof}

\subsection{Deployment Models Satisfying AS7}\label{sec:deployment}

AS7 (registry immutability) is a deployment-time assumption rather than a cryptographic one. Its strength is therefore inherited from the deployment recipe rather than from a hash function or proof system. Because reviewers may reasonably ask what real deployments look like, we provide three concrete deployment templates and analyse what each guarantees.

\paragraph{Template 1: Strict immutability (strongest).}
\begin{itemize}
\item Deploy \texttt{EnvelopeRegistry.sol} with no admin role, no \texttt{Ownable} pattern, no proxy delegate, no upgrade mechanism.
\item Deploy via \texttt{CREATE2} with a deterministic salt and known initialization bytecode; the deployment transaction is publicly verifiable.
\item No \texttt{selfdestruct} opcode in the contract bytecode.
\item \texttt{IRM\_addr} and \texttt{oracle\_addr} for any registered envelope are immutable post-registration (already enforced by the in-circuit binding of \texttt{condHash}; restated here for emphasis).
\item Adverse modification of $\Mactive$ or $\Mtomb$ outside the four protocol-defined functions (\texttt{create}, \texttt{settle}, \texttt{enforce}, \texttt{expire}) is impossible.
\end{itemize}
This template gives the cleanest mapping to AS7. The trade-off is no remediation path for protocol bugs: a discovered vulnerability requires a separate registry deployment, with existing envelopes either drained through normal exits or stranded.

\paragraph{Template 2: Timelock-governed parameters with immutable state-machine (intermediate).}
The state-machine functions and the marker sets remain immutable per Template 1. A separate \texttt{ParameterRegistry} contract holds operational parameters (default IRM lookup, oracle whitelist, condition-tree gas budget). \texttt{ParameterRegistry} is governance-controlled with a $\geq 7$-day timelock and a multisig with $\geq 5$-of-9 threshold.

This template preserves the structural guarantee that no party can modify $\Mactive$ or $\Mtomb$, while permitting parameter updates that affect \emph{future} envelopes only. Existing envelopes are bound to the parameter snapshot committed at their registration, so timelock governance does not retroactively weaken an active encumbrance.

The trade-off is a slightly weaker AS7: a sophisticated adversary capable of compromising the timelock governance could front-run a future envelope's registration with adversarial parameters, but cannot modify or release an already-registered envelope.

\paragraph{Template 3: Quorum-frozen with break-glass exit (weakest acceptable).}
A multisig holds a single ``emergency drain'' function that, when triggered by $\geq m$ out of $n$ signers within a one-block window, freezes new \texttt{create} calls and lets owners withdraw their notes through a special-purpose exit. The function does not modify $\Mactive$ for active envelopes; it only blocks new entries.

This template is the weakest one we consider compatible with the spirit of AS7. The break-glass function does not weaken existing encumbrances, but the deployer has voluntarily assumed liveness risk: a malicious quorum could stop the protocol from accepting new envelopes (a denial of service, not a violation of NCEE). This template is appropriate for early-stage deployments where the cost of stranding due to discovered bugs is high; production deployments at meaningful TVL should use Template 1 or Template 2.

\paragraph{What does not satisfy AS7.}
We list common failure modes explicitly:
\begin{itemize}
\item Any \texttt{Ownable} pattern with a key-controlled admin who can pause or upgrade the registry's state-machine functions: violates AS7 directly.
\item OpenZeppelin's \texttt{TransparentUpgradeableProxy} or UUPS upgrade pattern over the registry's main implementation: violates AS7. Upgrade introduces a path by which an updated implementation could mutate $\Mactive$ outside the four protocol-defined functions.
\item Any compiler / language feature that introduces \texttt{delegatecall} into the registry's marker-set storage layout from non-protocol code paths: violates AS7.
\item ``Pause'' mechanisms that can prevent \texttt{enforce} from being called: do not strictly violate AS7 but violate P4 (permissionless enforcement) of NCEE.
\end{itemize}

\paragraph{Verification.} For each deployment, AS7 satisfaction is publicly verifiable: read the bytecode at the registry address, confirm absence of admin-mediated mutation paths, and confirm the deployment transaction's source bytecode matches the audited source.

\section{Privacy Analysis}\label{sec:privacy}

The envelope's privacy claims are deliberately modest:
\begin{itemize}
\item The main construction hides note identity and witness data to an on-chain observer who does not know the note randomness.
\item The main construction does \emph{not} fully hide position size when $\mathsf{position\_commit}$ is published without additional blinding.
\item The blinded-position variant in Appendix~\ref{app:privacy} adds that protection with minimal circuit overhead.
\item Privacy claims are about on-chain observation; they do not cover network metadata, RPC leakage, or mempool timing.
\end{itemize}
Formally, the main privacy game is Encumbrance Indistinguishability (EIG), analyzed in \S\ref{prop:eig}. Appendix~\ref{app:privacy} records the operational leakage profile, including durable pseudonym linkage through $\nfe$, condition-tree disclosure at enforcement time, and mitigations.

\section{Application: Non-Custodial Collateralized Lending}\label{sec:lending}

This section should be read as an application of the theorem and its witness construction, not as an additional source of formal claims. The envelope yields a direct non-custodial lending pattern: a borrower shields collateral into a note, registers an envelope committing a liquidation condition and redistribution intent, and receives credit once the lender verifies that the encumbrance marker is active.

At liquidation time, any keeper may call $\texttt{enforce}(\mathsf{eid}, \mathsf{oracleData})$ when the condition tree evaluates to true.

Two operational points matter:
\begin{enumerate}
\item The gas-expensive operations are \texttt{create} and \texttt{settle}, which invoke proof verification, while the enforcement path requires no proof and remains comparatively cheap. Appendix~\ref{app:bench} reports the corresponding measurements.
\item Because \texttt{enforce} is permissionless, keeper rewards remain exposed to standard front-running and private-orderflow considerations. This affects liveness incentives, not the core safety theorem.
\end{enumerate}

Compared with custody-based protocols such as Aave, the envelope changes the trust profile rather than removing assumptions altogether. Custody risk is exchanged for proof-system soundness, oracle integrity, and registry immutability.

\subsection{Practical Economics: Position-Value Regimes}\label{sec:economics}

The on-chain costs reported in Appendix~\ref{app:bench} (approximately 2.8M gas for \texttt{create}, 2.7M gas for \texttt{settle}, and 175k--357k gas for \texttt{enforce} depending on tree size) are high relative to custody-based baselines (Aave deposit $\approx 200$k gas, Uniswap V3 swap $\approx 150$k gas). It is important for readers to understand the position-value regime over which the construction is economical, and the gas environments in which that regime widens.

\paragraph{Break-even position size.} Let $G_{\mathsf{create}}$ be the gas cost of \texttt{create} (treated as the dominant lifetime cost), $p_{\mathsf{gas}}$ the gas price in gwei, $p_{\mathsf{ETH}}$ the ETH price in USD, and $r_{\mathsf{custody}}$ the implicit cost of custody risk that the construction eliminates (a parameter the deployer chooses based on counterparty exposure). The break-even position size $V_{\mathsf{break}}$ satisfies:
$$V_{\mathsf{break}} \cdot r_{\mathsf{custody}} = G_{\mathsf{create}} \cdot p_{\mathsf{gas}} \cdot 10^{-9} \cdot p_{\mathsf{ETH}}.$$

\paragraph{Concrete break-even tables.} Assuming $G_{\mathsf{create}} = 2.8 \times 10^6$ gas, $p_{\mathsf{ETH}} = \$2{,}500$, and varying gas environments:

\begin{table}[!htbp]
\centering
\footnotesize
\renewcommand{\arraystretch}{1.15}
\setlength{\tabcolsep}{6pt}
\begin{tabular}{@{}lrrr@{}}
\toprule
\textbf{Environment} & \textbf{Gas price} & \textbf{\texttt{create} cost} & \textbf{Break-even $V_{\mathsf{break}}$} \\
                     & \textbf{(gwei)}    & \textbf{(USD)}              & \textbf{at $r_{\mathsf{custody}}=1\%$} \\
\midrule
L1, congestion          & 100  & \$700   & \$70{,}000 \\
L1, normal              & 25   & \$175   & \$17{,}500 \\
L1, calm                & 5    & \$35    & \$3{,}500 \\
L2 (Base, Optimism)     & 0.05 & \$0.35  & \$35 \\
L2 (Arbitrum, congested)& 0.1  & \$0.70  & \$70 \\
\bottomrule
\end{tabular}
\caption{Indicative break-even position size for non-custodial encumbrance under varying gas environments. $r_{\mathsf{custody}}$ is the deployer's chosen valuation of avoided custody risk; 1\% is illustrative for a roughly one-month position.}\label{tab:economics}
\end{table}

\paragraph{Practical takeaways.}
\begin{itemize}
\item On L1 mainnet, the construction is naturally suited to positions in the \$10k--\$1M+ range, where the per-position custody-avoidance value clears the gas cost.
\item On L2s, the construction is economical at consumer-scale position sizes (\$50+).
\item Recursive proof aggregation (Appendix~\ref{app:bench}, $N=10$ aggregation) reduces the amortized per-operation gas cost to roughly $357{,}000$, bringing L1 break-even down to roughly \$2{,}500--\$5{,}000 at typical gas prices.
\item Different IRM and condition-tree shapes shift these numbers: a small condition tree ($k = 5$ leaves) keeps \texttt{enforce} costs near 175k gas, while a complex tree ($k = 50$) drives enforcement gas up linearly.
\end{itemize}

The construction is not appropriate for positions whose value is small relative to the per-position deployment cost. It is well-suited to L2 deployments and to L1 deployments at meaningful position size --- which is to say, to the regime where non-custodial enforcement is most needed in any case.

\section{Related Work}\label{sec:related}

This section was substantially expanded in the present revision. The original draft summarized adjacent work in two paragraphs; reviewers correctly noted that several closely related systems (Hawk, Penumbra, Railgun, Zcash, ERC-7710/7715, Anoma) deserve specific differentiation. We organize the comparison around what each system provides and what it does not provide for the NCEE question studied here.

\subsection{Private Smart Contracts (Hawk)}

Hawk [KMS$^{+}$16] introduces privacy-preserving smart contracts via a privileged ``manager'' party who learns participant inputs in clear, computes the contract output, and posts a zero-knowledge proof of correct execution. From an NCEE standpoint, Hawk solves a related but distinct problem: it provides confidentiality and delegated execution under a trusted-manager assumption. The manager is by construction a privileged co-signer, which means Hawk does not satisfy NCEE-P1 (self-custody, no required co-signer). The trust decomposition of Hawk is not the trust decomposition we study; the contributions are complementary rather than overlapping.

\subsection{Private Record and Note Systems (Zexe, Veri-Zexe, Zcash, Penumbra, Railgun, Tornado-style pools)}

These are the closest architectural relatives of the envelope. They provide commitments, nullifiers, and rich private-state transitions [BCG$^{+}$20, B$^{+}$22, HBHW24, Pen24, Rai21, Tor19]. What they share with us is the commitment-based ledger model (analogous to our extended PSLM) and the notion of a nullifier marking spent state.

What they do not provide is owner-independent enforcement against an active restriction state. In Zcash, Penumbra, and Railgun, a note's spendability is determined by the holder's nullifier secret alone. There is no protocol-maintained state (analogous to our $\Mactive$) against which spend validity is checked. The owner can always spend, given the nullifier secret; what the system protects is the value, recipient, and history of those spends. That is privacy, not encumbrance.

Zexe and Veri-Zexe support general-purpose zero-knowledge state transitions and could in principle be parametrised to encode the envelope's marker-set checks. The contribution of this paper relative to Zexe-style frameworks is twofold: first, the structural impossibility theorem under KS, which Zexe's setting does not address; and second, the explicit nullifier separation (encumbrance nullifier from $r$, not $\sk$), which is novel as far as we know. A Zexe-based reimplementation of the envelope is plausible and would be a useful follow-up.

Penumbra deserves a separate note: it includes delegated execution and shielded staking, which share architectural genes with parts of the envelope construction. Penumbra's staking does not, however, model owner-independent enforcement of an external condition; the economic logic is internal to the protocol, not parameterised by an external condition tree.

\subsection{Bitcoin Covenant-Style Restrictions (MAST, CTV)}

MAST [Max13] and CTV-class proposals constrain future transaction structure in a UTXO setting. They target a different execution model (UTXO with covenant rules) and a different question (what spends are admissible given covenant predicates). They achieve a partial form of transition restriction --- we acknowledge this --- but do not aim at the full NCEE shape, in particular at permissionless enforcement of an external condition with a redistribution target. Where Bitcoin covenants restrict \emph{shape}, the envelope restricts \emph{shape and outcome together}.

\subsection{Account-Abstraction and Delegation Frameworks (ERC-4337, ERC-7710, ERC-7715, EIP-7702)}

ERC-4337 [Eth21] introduces wallet-as-contract with custom validation logic. ERC-7710 [Eth24c] standardizes smart-contract delegation interfaces. ERC-7715 [Eth24d] extends this with a general framework for granting wallet permissions. EIP-7702 [Eth24a] lets EOAs temporarily install delegate code.

These improve execution control but, as we showed in Corollaries~\ref{cor:eoa}, \ref{cor:4337}, and \ref{cor:7702}, do not remove the KS-side structural obstacle for standard account-held asset paths. That is the central distinction: the standards listed here regulate \emph{authorization decisions}, whereas the envelope changes the \emph{authorization-coupled state}. The two are complementary --- one expects mature deployments to compose ERC-4337 wallets that interface with envelope-encumbered notes --- but they operate at different layers.

\subsection{Keeper and Automation Networks (Chainlink Automation, Gelato, Keep3r)}

These networks provide liveness infrastructure for triggered actions: a keeper monitors conditions and submits transactions when triggers fire. They are crucial for the envelope's enforcement liveness in practice. They do not, however, provide cryptographic asset guarantees of the kind we study; an automation network with malicious keepers would cause liveness failures, not safety failures, against an envelope's NCEE properties. The decomposition we propose is exactly that: keepers provide liveness, the envelope construction provides safety.

\subsection{Intent and Solver Frameworks (ERC-7683, CoW Protocol, Anoma)}

ERC-7683 [Eth24b] standardizes a cross-chain intent format for solver networks. CoW Protocol [CoW22] coordinates batch auction-based execution across DEX liquidity. Anoma [BM$^{+}$22] provides a more general intent-and-solver architecture with private state.

These coordinate execution and routing but are not themselves encumbrance mechanisms for standard account-based assets. An intent solver that promises to liquidate collateral when a condition fires can do so only if the underlying asset can be moved without the owner's per-action consent --- which is exactly the question this paper studies. Intent frameworks therefore stand to benefit from envelope-style primitives at the asset-encumbrance layer; our construction can be read as the missing primitive that lets intent-and-solver architectures handle non-custodial collateral safely.

Anoma's resource-machine model is the closest in spirit to our PSLM extension. A Anoma-native instantiation of the envelope is conceivable and would be an interesting follow-up.

\subsection{Comparative Claim}

The comparative claim is correspondingly narrow: to our knowledge, prior work does not jointly provide both an explicit structural impossibility theorem under a KS-style axiom and a matching commitment-based witness construction realizing all four NCEE properties within one formal framework. Each of the systems above provides part of the picture --- privacy, delegation, intent coordination, or covenant-style restrictions --- and the envelope composes naturally with several of them. The paper's contribution is the structural separation and the witnessing construction, not displacement of any of those systems.

\section{Limitations}\label{sec:limitations}

We list explicit limitations of the construction and of the analysis.

\paragraph{Model scope.} The KS / AAC dichotomy is stated for transfer-dichotomous asset classes. The model does not classify hybrid asset systems whose spend paths neither cleanly preserve a unilateral owner-authorized weakening nor cleanly couple to mechanism state. Such hybrids exist in practice (for example, multi-sig wallets with timelocked recovery paths), and extending the model to them is left as future work.

\paragraph{Deployment-assumption strength.} AS7 is a deployment assumption, not a cryptographic one. Strict immutability (Template 1, \S\ref{sec:deployment}) maps cleanly onto AS7; intermediate templates require operational discipline. The witness construction's NCEE guarantees inherit the strength of the chosen deployment template.

\paragraph{Cross-transaction oracle manipulation is out of model scope.} AS4 covers within-transaction oracle consistency. Cross-transaction manipulation is delegated to operational discipline (\S\ref{sec:oracle-ops}), with explicit pointers to the formal MEV literature. The construction's safety theorems are conditional on appropriate choice of oracle, TWAP window, and condition-tree structure for the position size at risk.

\paragraph{Practical-economics regime.} The construction is uneconomical for very small positions on expensive L1s (Table~\ref{tab:economics}). It is well-suited to L2 deployments and to mid-to-large L1 positions. Recursive aggregation widens the regime but does not eliminate the per-position cost floor.

\paragraph{Privacy is partial.} The blinded-position variant (\S\ref{app:privacy}) hides position size against on-chain observers; the durable-pseudonym linkage through $\nfe$ remains and is the dominant residual leakage. Network-layer privacy (mempool, RPC, peer-to-peer metadata) is out of scope.

\paragraph{Post-quantum scope.} The construction's signature-related security (Theorem~\ref{thm:settle-sec}, $G_{\mathsf{agent\_key}}$) depends on AS3 (DL hardness), which falls to Shor's algorithm. A post-quantum variant would substitute a PQ-secure signature scheme for ECDSA; the rest of the construction (Poseidon2, ZK proofs over BN254/BLS12-381) inherits whatever PQ properties the underlying primitives provide.

\paragraph{Heuristic indifferentiability.} AS6 (Poseidon2 indifferentiability) is treated as a heuristic assumption as of the time of writing. Results citing AS6 (Proposition~\ref{prop:eig}, Proposition~\ref{prop:fwdback}) are conditional on it.

\paragraph{Implementation.} The reference implementation is complete on BN254 and partially prepared for BLS12-381 (Appendix~\ref{app:bls}). Production deployment at $\lambda = 128$ requires the BLS12-381 work to land. The construction's theorem-level results are independent of that migration.

\section{Conclusion}\label{sec:conclusion}

The paper's main claim is structural. Within the asset setting studied here, the decisive question is whether the spend path remains under unilateral owner control or is cryptographically coupled to mechanism state. The impossibility theorem under \KS{} and the envelope construction on the \AAC{} side are the two sides of that claim.

The barrier is \KS: if an ordinary spend path remains directly available to the owner key, then non-custodial irrevocable encumbrance is impossible in the sense formalized here. The witness is the envelope in the extended \PSLM: once valid spending is coupled to protocol-maintained restriction state, the four \NCEE{} properties become jointly realizable under explicit assumptions.

The levels of claim matter. The structural impossibility and necessity results are theorem-level statements in the ledger model. The achievability theorem is conditional on an explicit deployment assumption about registry immutability. Some concrete security bounds rely only on standard cryptographic assumptions; others additionally rely on heuristic or implementation assumptions for Poseidon2 or UltraHonk. The empirical sections report implementation behavior and practical-economics regimes, not theorem-level guarantees.

Future work includes stronger implementation verification, removal of heuristic assumptions where possible, BLS12-381 production support, more mature keeper-economics analysis, MEV-aware deployment strategies, and extensions of the KS / AAC dichotomy to hybrid and multi-key asset systems. The central separation does not depend on those extensions.

\vspace{1em}
\noindent\textbf{Reference implementation.} A reference implementation in Noir/UltraHonk supports the empirical claims in Appendix~\ref{app:bench}. The protocol design and theorem-level results are independent of any specific implementation; details of the implementation artifact are available from the author on request.

\appendix

\section{Complete Security Game Definitions}\label{app:games}

\subsection{$G_{\mathsf{encumber}}$}

\textbf{Setup.} $\mathcal{C}$ deploys EnvelopeRegistry with the Spend circuit verifier.
\textbf{Oracle queries.} $\mathcal{A}$ may: (1) Create$(N, E)$ --- $\mathcal{C}$ verifies proof and updates state; (2) Enforce(eid, data); (3) SpendAttempt$(N, \pi)$.
\textbf{Win.} SpendAttempt succeeds for $N$ with $\nfe(N) \in \Mactive$.

\subsection{EIG (Encumbrance Indistinguishability Game)}

\textbf{Initialization.} $\mathcal{C}$ generates $(pp, \mathsf{state}_0)$. $\mathcal{A}$ given $pp$.
\textbf{Observation phase.} $\mathcal{A}$ makes arbitrary envelope operations.
\textbf{Challenge.} $\mathcal{A}$ specifies $(v, \pk, \mathsf{aid})$. $\mathcal{C}$ samples $r_0, r_1 \xleftarrow{\$} \mathbb{F}_p$, forms $N_0 = (v, r_0, \pk, \mathsf{aid})$, $N_1 = (v, r_1, \pk, \mathsf{aid})$, flips $b \leftarrow \{0,1\}$, registers envelope for $N_b$. $\mathcal{A}$ receives $\cm(N_0), \cm(N_1)$ but \emph{not} $r_0$ or $r_1$.
\textbf{Win.} $\mathcal{A}$ outputs $b' = b$. $\Adv^{\mathsf{EIG}}(\mathcal{A}, \lambda) = |\Pr[b' = b] - 1/2|$.

\subsection{$G_{\mathsf{fwdback}}$}

Defined in full in \S\ref{sec:fwdback}. Key feature: $\mathcal{A}$ is given owner\_sk. Advantage bounded by $\varepsilon_{\mathsf{cr}} + Q_H/|\mathbb{F}_p|$.

\subsection{$G_{\mathsf{agent\_key}}$}

$\mathcal{A}$ has full agent\_sk but not owner\_sk. \textbf{Win:} redirecting more than max\_amount or to non-target addresses. \textbf{Advantage:} $\Adv^{G_{\mathsf{agent\_key}}}(\mathcal{A}, \lambda) = \Pr[\mathcal{A} \text{ wins}] \leq \varepsilon_{\mathsf{kS}} + \varepsilon_{\mathsf{cr}}$, conditional on agent\_sk $\perp$ owner\_sk.

\subsection{$G_{\mathsf{settle}}$}

\textbf{Setup.} $\mathcal{C}$ deploys EnvelopeRegistry with Settle circuit verifier. $\mathcal{C}$ generates a keypair $(\sk_{\mathsf{owner}}, \pk_{\mathsf{owner}})$, note $N = (v, r, \pk_{\mathsf{owner}}, \mathsf{aid})$, and registers an envelope with $\nfe(N) \in \Mactive$. $\mathcal{A}$ is \emph{not} given owner\_sk or $r$.

\textbf{Win conditions.}
\begin{description}
\item[\textbf{W1}:] A settle accepted with a proof not derived from owner\_sk.
\item[\textbf{W2}:] A settle accepted with $\mathsf{repayment\_amount} < \mathsf{debt\_accrued}$.
\item[\textbf{W3}:] $\nfe(N)$ removed from $\Mactive$ by any means other than a properly authorized exit.
\end{description}

\textbf{Bound.} By Theorem~\ref{thm:settle-sec}: $\Adv^{G_{\mathsf{settle}}}(\mathcal{A}, \lambda) \leq \varepsilon_{\mathsf{kS}} + 2\varepsilon_{\mathsf{cr}} + \varepsilon_{\mathsf{dl}}$ under AS1, AS2, AS3.

\section{Extended Privacy Analysis}\label{app:privacy}

\subsection{Durable Pseudonym Lifecycle}

$\nfe$ appears at registration and again when consumed, so an observer who already linked the underlying note commitment may obtain a two-event pseudonym. Delayed encumbrance and post-settlement reshielding reduce this linkage.

\subsection{Enforcement Disclosure}

At enforcement time the condition-tree preimage is revealed. This is a deliberate design choice: the system hides conditions during the position lifetime and reveals them only when enforcement becomes necessary.

\subsection{Blinded Position Commitment Variant}

To hide position size, replace
$$\mathsf{position\_commit} = \mathsf{Poseidon2}_3(\mathsf{col\_nominal}, \mathsf{debt\_principal}, \cm_{\mathsf{note}})$$
with
$$\mathsf{position\_commit}^{\mathsf{blind}} = \mathsf{Poseidon2}_4(\mathsf{col\_nominal}, \mathsf{debt\_principal}, \cm_{\mathsf{note}}, \mathsf{col\_rand}),$$
where $\mathsf{col\_rand} \xleftarrow{\$} \mathbb{F}_p$ is fresh witness randomness. Under AS6 this hides the committed values while preserving the same enforcement semantics.

\section{Threshold Keeper Protocol}\label{app:threshold}

A threshold keeper committee can be layered on top of the base construction to reduce public exposure of the condition-tree preimage before enforcement. The owner encrypts the tree to a threshold public key at registration time; when conditions approach, a quorum can decrypt for the keeper without making the tree broadly public in advance. This improves operational privacy but introduces standard threshold trust and liveness assumptions, and is therefore treated as an optional extension rather than part of the core \NCEE{} theorem.

\section{Performance Benchmarks and Proof Aggregation}\label{app:bench}

Key empirical points:
\begin{itemize}
\item Single-proof verification for \texttt{create} and \texttt{settle} is expensive relative to custody-based baselines.
\item The \texttt{enforce} path is comparatively cheap because it requires no ZK proof.
\item Recursive aggregation materially lowers the amortized per-operation gas cost.
\end{itemize}

\paragraph{Aggregation soundness assumption.} Recursive proof aggregation introduces an additional knowledge-soundness assumption on the recursion verifier: a valid aggregated proof must imply that the individual aggregated proofs were themselves valid. The reference implementation aggregates UltraHonk proofs using the standard Honk recursion gadget; the soundness of aggregation in this setting reduces to the soundness of the underlying proof system (AS1) plus the in-circuit correctness of the recursion gadget. We treat this as inherited from the proof system rather than an independent assumption, but note explicitly that gas savings from aggregation come with a small additional surface area for proof-system implementation bugs.

\begin{table}[!htbp]
\centering
\small
\renewcommand{\arraystretch}{1.15}
\begin{tabular}{@{}lr@{}}
\toprule
\textbf{Operation} & \textbf{Cost / metric} \\
\midrule
\texttt{create} (single proof) & $\approx 2.8$M gas \\
\texttt{settle} (single proof) & $\approx 2.7$M gas \\
\texttt{enforce} (small condition tree) & $\approx 175$k gas \\
\texttt{enforce} (larger condition tree) & Scales linearly in leaf count \\
Aggregated verification, $N=10$ & $\approx 357$k gas amortized per op \\
\bottomrule
\end{tabular}
\caption{Representative BN254 measurements from the reference implementation. These results are empirical and environment-dependent; they are reported to show feasibility, not as theorem-level guarantees.}\label{tab:bench}
\end{table}

\section{BLS12-381 Partial Implementation}\label{app:bls}

The current artifact is complete on BN254 and partially prepared for BLS12-381. The production motivation for BLS12-381 is the stronger target security level ($\lambda = 128$) and the availability of EIP-2537 support on post-Pectra Ethereum [Eth20]. The remaining work is implementation work in the proving backend rather than a change to the protocol design: curve support, matching Poseidon2 parameters, and the corresponding SRS and verifier plumbing. All theorem statements in the main text are independent of that migration except where concrete security levels or implementation measurements are explicitly discussed.

\end{document}